\documentclass[11pt]{amsart}
\usepackage{geometry}                
\usepackage{graphicx}
\usepackage{amssymb}
\usepackage{epstopdf}
\usepackage[backend=bibtex,style=alphabetic]{biblatex}
\newcommand{\eff}{\mathrm{eff}}
\newcommand{\R}{\mathbb{R}}
\newcommand{\rem}[1]{}

\renewcommand{\a}{\boldsymbol{a}}
\renewcommand{\l}{\boldsymbol{l}}
\newcommand{\e}{\boldsymbol{e}}

\newcommand{\z}{\zeta}

\newtheorem{theorem}{Theorem}

\title{The Lagrange top and the fifth Painlev\'e equation}
\author{Holger R.~Dullin}
\address{School of Mathematics and Statistics \\ University of Sydney}
\email{holger.dullin@sydney.edu.au}
\date{}                                           

\begin{document}
\begin{abstract}
We show that the Lagrange top 
with a linearly time-dependent moment of inertia is equivalent to 
the degenerate fifth Painlev\'e equation. More generally we show that the harmonic Lagrange top
(the ordinary Lagrange top with a quadratic term added in the potential) is equivalent to 
the fifth Painlev\'e equation when the potential is made time-dependent in an appropriate way.
Through this identification two of the parameters of the fifth Painlev\'e equation
acquire the interpretation of global action variables. We discuss the relation to the confluent 
Heun equation, which is the Schr\"odinger equation of the Lagrange top, and discuss the dynamics
of $P_V$ from the point of view of the Lagrange top.
\end{abstract}

\maketitle

\section{Introduction}

The Painlev\'e equations are six non-linear second order ODEs, all of whose moveable singularities are poles.
They were initially studied by P.~Painlev\'e, B.~Gambier, R.~Fuchs and others around 1900, and today 
are at the centre of the theory of integrable systems.
For a general introduction see \cite{Iwasaki91,Gromak02,noumi04,Fokas06,conte19}.
Painlev\'e equations appear in the Ising model \cite{wu76}, 
plasma physics \cite{hastings80},
Bose gas \cite{jimbo80},
random matrix theory \cite{tracy94}, 
as reductions of integrable PDEs \cite{ablowitz03}, 
and we refer to \cite{conte19} for a more extensive list of applications.
The fifth Painlev\'e equation, denoted by $P_V$, for $w = w(\z)$ is
\begin{equation} \label{eqn:PVw}
\frac{{\mathrm{d}}^{2}w}{{\mathrm{d}\z}^{2}}=
\left(\frac{1}{2w}+\frac{1}{w-1}\right)\left(\frac{\mathrm{d}w}{\mathrm{d}\z}\right)^{2}
-\frac{1}{\z}\frac{\mathrm{d}w}{\mathrm{d}\z}
+\frac{(w-1)^{2}}{\z^2}\left(\alpha w+\frac{\beta}{w}\right)
+\frac{\gamma w}{\z}
+\frac{\delta w(w+1)}{w-1}
\end{equation}
where $\alpha, \beta, \gamma, \delta$ are constants, see,~e.g.,~\cite{DLMF}.

In this paper we would like to add the Lagrange top to the list of applications: 
The fifth  Painlev\'e equation $P_V$ describes the symmetric rigid body with a fixed point in a quadratic potential,
i.e.\ the harmonic Lagrange top of \cite{DDN22}, with a time-dependent potential. 
Furthermore, the usual Lagrange top in the linear potential of gravity with a moment of inertia depending linearly on time is equivalent to $P_V$ with $\delta = 0$, the so called degenerate case of $P_V$. 
In a somewhat similar spirit a connection between a non-autonomous Euler top with extra gyroscopic terms and $P_{VI}$ has been reported in \cite{LevinOlshanetskyZotov06}.

In \cite{DDN22} we showed that the quantisation of the harmonic Lagrange top (i.e.\ a symmetric top in a quadratic potential) leads to a 
Schr\"odinger equation which is the confluent Heun equation. 
In \cite{Slavyanov96,Slavyanov00} it was shown that Heun equations 
are related to Painlev\'e equation by a kind of de-quantisation procedure.
In fact the relation between the Heun equation and the Painlev\'e equation was classically known, for some modern references see
\cite{Fokas06,Dubrovin18,Lisovyy21}.
This motivated the idea that the harmonic Lagrange top when appropriately turned into a non-autonomous system is equivalent to $P_V$. Here we are going to show that this is indeed the case.
We directly establish the equivalence of $P_V$ and the non-autonomous (harmonic) Lagrange top
without the detour through the Heun equation, but will comment on the connection to the Heun equation in a later section.
In the two final section we consider regulariations of the singular points $w = 0, \infty$ motivated through the Lagrange top.
After symplectic reduction by one $S^1$ symmetry the dynamics of the Lagrange top lives on $T^*S^2$ and gives
a singularity free description of the dynamics of $P_V$ on $S^2$, and also a simple qualitative description of real solutions of $P_V$. 
In the final section we consider the 
full singular symmetry reduction by $S^1 \times S^1$ which leads to dynamics on an orbifold.
Both these description could be considered as a kind of blow-up of $P_V$.

\section{Trigonometric form of $P_V$}

Changing the independent variable to $\tau = \log \z$ 
and redefining the constants according to 
$\kappa_\infty^2 = 2 \alpha$, $\kappa_0^2 = -2 \beta$
gives the modified fifth Painlev\'e equation  \cite{Gromak02} as
\[
\frac{{\mathrm{d}}^{2}w}{{\mathrm{d}\tau}^{2}}=
\left(\frac{1}{2w}+\frac{1}{w-1}\right)\left(\frac{\mathrm{d}w}{\mathrm{d}\tau}\right)^{2}
+ \frac12 (w-1)^{2} \left(\kappa_\infty^2 w-\frac{\kappa_0^2}{w}\right)
+ \gamma e^\tau w
 + \delta e^{2\tau} \frac{w(w+1)}{(w-1)} \,.
\]
This equation has the property  that every solution is locally meromorphic \cite{Joshi94,Hinkkanen01}.
The new form of the parameters is convenient for discussion of the affine Weyl symmetry group $W(A_3^{(1)})$ \cite{Okamoto87},
and in particular also for the description of special function solutions and rational solutions \cite{Kitaev94,Kajiwara02,Gromak02,Umemura96,Clarkson05}.

The first polynomial Hamiltonian form of $P_V$ was given by Okamoto  \cite{Okamoto87}.
A Hamiltonian form of $P_{VI}$ in which the Hamiltonian has the standard form $H = \tfrac12 p^2 + V(q)$ 
was given by Manin \cite{Manin96,Manin05} where $V$ is given in terms of the Weierstra\ss{} $\wp$ function,
although the corresponding form of $P_{VI}$ was already described in slightly different form by Fuchs \cite{Fuchs1905} and Painleve \cite{Painleve1906}.
This and analogous transformations for other Painlev\'e equations were given by Babich and Bordag \cite{BabichBordag99}, 
Iwasaki \cite[4.2.1]{Iwasaki91}, Takasaki \cite{Takasaki01}, also see \cite{LevinOlshanetsky00}.
Introducing a new dependent variable by $w = \coth^2 y/2$ transforms the modified $P_V$ equation
into the hyperbolic form
\[
  \frac{d^2 y}{d \tau ^2} = -V', \quad 
  V(y) = -\frac{\kappa_\infty^2}{ 2 \sinh^2 (y/2)} + \frac{\kappa_0^2}{2 \cosh^2 (y/2)} + \frac{ \gamma e^\tau }{2} \cosh y
         + \frac{ \delta e^{2\tau} }{4} \cosh^2 y \,.
\]
In order to obtain an equation related to the Lagrange top instead we consider the slightly different transformation $w = -\cot^2 y/2$, which leads to 
\begin{equation} \label{eqn:trigform}
   \frac{d^2 y}{d\tau^2} = -V', \quad 
  V(y) = -\frac{\kappa_\infty^2}{2 \sin^2(y/2)} - \frac{\kappa_0^2}{2 \cos^2(y/2)} - \frac{ \gamma e^\tau }{2} \cos y
         - \frac{ \delta e^{2\tau} }{4} \cos^2 y \,.
 \end{equation}
 We call this equation the trigonometric form of $P_V$.
 It is obtained from the hyperbolic form by the simple transformation $y \to i y$.
 A frozen time version of this equation is obtained by setting $\tau=0$ in the exponential terms, and in frozen time this is the equation for the harmonic Lagrange top, as we are now going to show.

\section{The Lagrange top}

The  Lagrange top is a symmetric heavy rigid body with a fixed point on the symmetry axis. The configuration space is $SO(3)$.
In Euler angles $\phi, \theta, \psi$ it has a metric on $SO(3)$ defined by the kinetic energy, see, e.g., \cite{LanLif84}, as
\begin{equation} \label{eqn:Tkin}
   T_\mathrm{rot} = \tfrac12 I_1 ( \dot \phi^2 \sin^2 \theta + \dot\theta^2) + \tfrac12 I_3 ( \dot \phi \cos\theta + \dot\psi)^2 \,,
\end{equation}
where $\phi$ and $\psi$ are $2\pi$-periodic angles and $\theta \in [0, \pi]$, and $I_1=I_2$ and $I_3$ are the principal moments of inertia of the body with respect to the fixed point.
A Legendre transformation leads to the corresponding Hamiltonian 
\begin{equation} \label{eqn:HLag}
    H =  \frac{1}{2 I_1} \left( p_\theta^2 + \frac{1}{\sin^2\theta} ( p_\phi^2 + p_\psi^2 - 2 p_\phi p_\psi \cos\theta) \right) + 
    \frac12 \left( \frac{1}{I_3} - \frac{1}{I_1} \right) p_\psi^2 + U( \cos\theta)
\end{equation}
where the potential $U$ depends on $z = \cos\theta$, the spatial $z$-coordinate of the tip of the axis of the top.
The usual Lagrange top in the field of gravity has only a linear term proportional to $z$ in the potential. 
The harmonic Lagrange top studied in \cite{DDN22} adds a  quadratic term and hence we consider $U(z) = c z + d z^2$.
The potential is left somewhat general as a function $U$ because we will later also allow for time-dependence in $U$.
Both momenta $p_\phi$ and $p_\psi$ are constants of motion, since the angles $\phi$ for rotation about
the direction of gravity and $\psi$ for rotation about the symmetry axis of the body are both cyclic.
The kinetic energy in the above Hamiltonian is split into a kinetic term that corresponds to the ``round'' top with all moments of inertia equal to $I_1$,
and an asymmetry ``correction" proportional to the angular momentum for rotation about the symmetry axis of the body $p_\psi^2$.
This correction term is irrelevant for the dynamics of $\theta$.

In the Lagrange top with time-dependent moments of inertia and/or time-dependent potential the momenta $p_\phi$ and $p_\psi$ are still constants of motion.
Thus the essential dynamics is given by a (singularly) reduced one degree of freedom system in which the momenta $p_\phi$ and $p_\psi$ are parameters and all the terms but $p_\theta^2$ are considered as the effective potential of the reduced system
\begin{equation}
     H = \frac{1}{2 I_1} p_\theta^2 + U_{\eff} (\cos\theta; p_\phi, p_\psi) \,.
\end{equation}
The  angles $\phi$ and $\psi$ are driven by the dynamics of $\theta$ through Hamilton's equation
\begin{equation} \label{eqn:phipsi}
    \frac{d \phi}{d t}  = \frac{p_\phi - 2 p_\psi \cos\theta}{I_1 \sin^2\theta} , \quad
    \frac{d \psi}{d t} = \frac{p_\psi - 2 p_\phi \cos\theta}{I_1 \sin^2\theta}  + \left( \frac{1}{I_3} - \frac{1}{I_1}\right) p_\psi \,.
\end{equation}

The Lagrange top (without time-dependent terms) is Liouville integrable with integrals $H=E$, $p_\phi$, $p_\psi$. 
The typical motion is quasiperiodic on 3-dimensional tori in phase space. 
In this motion the tip of the axis of the top oscillates between $\theta_{min}$ and $\theta_{max}$ determined 
by $p_\phi$, $p_\psi$, and $E$, while rotating about its axis.
The constants of motion $p_\phi$ and $p_\psi$ are global action variables, they generate $2\pi$-periodic flows which 
are the rotation about the axis of gravity and the rotation about the axis of symmetry of the top, respectively. 
The third action variable is given by a complete elliptic integral of 3rd kind. 
Solutions on 2-dimensional tori occur for $\theta = const$ in which the tip of the axis of the top traces out a horizontal circle.
Isolated periodic solutions are the so-called sleeping tops with $\theta = 0$ (upright) or $\theta = \pi$ (hanging) where the axis of symmetry is parallel to the direction of gravity and the top is rotating about this axis. The sleeping tops are only possible for $p_\phi  \pm p_\psi = 0$, so that the term in the Hamiltonian that is singular for $\theta \to 0$ or $\theta \to \pi$, respectively, disappears. These linear combinations of $p_\phi$ and $p_\psi$ will play an essential role in the following. Finally, for $p_\phi = p_\psi = 0$ 
there are two equilibrium points corresponding to minimal and maximal potential energy.

\section{The equivalence between $P_V$ and the Lagrange top}

Now the stage is set to show that the two dynamical systems described in the previous two sections 
are actually equivalent with the appropriate choice of variables, parameters, and potentials.

\begin{theorem}
The trigonometric form of $P_V$ is the equation of motion for the harmonic Lagrange top where 
$y=\theta$, $\kappa_0^2 = -(p_\phi+p_\psi)^2/4$, $\kappa_\infty^2 = -(p_\phi-p_\psi)^2/4$, $\tau = t/I_1$,
 and $U$ is the time-dependent potential $U(z) = -( \tfrac12 \gamma e^\tau z + \tfrac14\delta e^{2\tau} z^2)/I_1$.
\end{theorem}

\begin{proof}
Consider the metric of the round $SO(3)$ of the rigid body with a fixed point given by 
\[
      \frac{1}{I_1} ds^2 = d\theta^2 + d\phi^2 + d\psi^2 + 2 \cos\theta d\phi d\psi
\]
obtained from the kinetic energy \eqref{eqn:Tkin} for $I_3 = I_1$.
This is a metric of constant sectional curvature $ 3/(2 I_1)$
whose Ricci tensor is proportional to the metric with proportionality factor $1/(2 I_1)$. Hence up to a covering it is equivalent 
to the metric of the round sphere $S^3$. To make this explicit introduce
new angles  $\phi_\pm$  through  $\phi_\pm = \phi \pm \psi$. 
In these coordinates the metric becomes diagonal
\[
    \frac{1}{I_1} ds^2 
    =  d\theta^2 +  \cos^2 \tfrac{\theta}{2}  d \phi_+^2  +  \sin^2 \tfrac{\theta}{2} d\phi_-^2  \,,
\]
and this is the metric of the Hopf coordinates on the sphere $S^3$ with angles $\phi_\pm$.
Note that at the coordinate singularity of the Euler angles where $\theta = 0$ only $\phi_+$ is defined,
while at $\theta = \pi$ only $\phi_-$ is defined.
Extending this to a symplectic transformation the momenta are given by $2 p_\pm = p_\phi \pm p_\psi$ and transforming \eqref{eqn:HLag} the new Hamiltonian is
\begin{equation} \label{eqn:Ham}
   H = \frac{1}{2I_1} \left( p_\theta^2 + \frac{p_+^2}{ \cos^2\theta/2} + \frac{p_-^2}{\sin^2\theta/2} \right) + 
    \frac12 \left( \frac{1}{I_3} - \frac{1}{I_1} \right) (p_+-p_-)^2
     +U(\cos\theta) \,.
\end{equation}
The overall factor $1/I_1$ can be removed by introducing a new time $\tau = t/I_1$.
The term proportional to $(p_+-p_-)^2$ has no influence on the dynamics of $\theta$ 
and can  be ignored. Thus define 
\[
   U_{\eff}(\cos\theta) = \frac{p_+^2}{2 \cos^2\theta/2} + \frac{p_-^2}{2 \sin^2\theta/2} + I_1 U(\cos\theta)
\]
as the effective potential relevant for the dynamics of $\theta(\tau)$.
Now Hamiltons equations for $\theta$ are equivalent 
to the trigonometric form \eqref{eqn:trigform} of $P_V$ in $y$ if we set
$V = U_\eff$ and hence the parameters in the effective potential are
$\kappa_0^2 = -p_+^2$, $\kappa_\infty^2 = -p_-^2$ 
and the coefficients in the potential $U(z) = c z + d z^2$ need to be chosen as
$ c = -\tfrac12 \gamma e^\tau / I_1$ and 
$ d = -\tfrac14 \delta e^{2\tau} / I_1$.
\end{proof}
The parameters $p_\pm$ are action variables and are therefore real for the Lagrange top, 
and hence the parameters $\kappa_0$, $\kappa_\infty$ in $P_V$ will be purely imaginary.
In particular this means that any rational solutions that appear for integer
or half-integer values of $\kappa_0$, $\kappa_\infty$, see, e.g., \cite{Kitaev94,Umemura96,Clarkson05}, are not relevant for the real Lagrange top,
similarly for special function solutions.
The transformation $w \to 1/w$ does map $P_V$ into itself with changed
parameters $(\alpha, \beta, \gamma) \to (-\beta, -\alpha, -\gamma)$.
However in terms of the signed parameters this becomes
$(\kappa_0^2, -\kappa_\infty^2, \gamma) \to (\kappa_\infty^2, -\kappa_0^2, -\gamma)$ 
and so is not able to flip the signs of $\kappa_0^2$, $\kappa_\infty^2$.
The only rational solution that does exists is the seed solution for 
B\"acklund transformations $w = -1$ for $\alpha + \beta = 0$ and $\gamma = 0$.
This is an equilibrium point of the potential $\delta \cos^2\theta$ at $\theta = \pi/2$.
The other two equilibrium points at $\theta = 0, \pi$  correspond to the singularities $w \to -\infty$ and $w \to 0$ in $P_V$, respectively.

The transformation of the metric to diagonal form suggest that another natural 
identification of $P_V$ can be made with the degenerate  Carl Neumann system on $T^*S^3$, see \cite{DH03}, where either the size of the sphere and/or the potential is time-dependent.

A different time-dependence for the Lagrange top is achieved by changing the moments of inertia, which is used to great effect, e.g., by figure skaters, and the next theorem is about this time-dependence. Note, however, that the figure skater mainly changes the moment of inertia $I_3$ about the axis of symmetry, which by way of \eqref{eqn:phipsi} will change the dynamics of $\psi$, the angle of rotation about that axis. Typically there will also be a small change in the moment of inertia $I_1$, and it is the time-dependence of $I_1$ that changes the dynamics of $\theta$, and thus gives the correspondence with $P_V$.

\begin{theorem}
The trigonometric form of the degenerate $P_V$  equation where $\delta = 0$ is the equation of motion for the Lagrange top with time-dependent moment of inertia $I_1(t) = a + b t$.
\end{theorem}

\begin{proof}
In this case the potential is simply $U = g \cos\theta$.
The proof proceeds as in Theorem~1 until the time is scaled.
In order to remove the time-dependent moment of inertia $I_1(t)$ from the kinetic energy introduce a new time by 
$dt = I_1(t) d\tilde \tau$. Now let $I_1(t) = a + b t$ and integration gives $\log(a + b t) = b(\tilde \tau - \tau_0)$ and hence 
$I_1(t) = a e^{b\tilde \tau}$. Finally define $\tau = b \tilde \tau$ and the Hamiltonian
\[
   H = \frac{1}{2} p_\theta^2 + U_\eff(\cos\theta), \quad
   U_\eff(\cos\theta) =  \frac{p_+^2}{2b^2 \cos^2\theta/2} + \frac{p_-^2}{2 b^2 \sin^2\theta/2} +  e^{\tau} \gamma \cos\theta
\]
where $\gamma = a g / b^2$ is that of the degenerate $P_V$ equation.
Transforming back to the original time $t$ we see that the Hamiltonian of the Lagrange top
in which $I_1(t) = a + b t$ directly gives the degenerate $P_V$ equation in the original time $t$.
\end{proof}

\section{The connection to the confluent Heun equation}

The confluent Heun equation written in the self-adjoint form (known as the generalised spheroidal wave equation) is
given by the linear 2nd order differential operator
\begin{equation} \label{eqn:Heun}
   L_{CH} = - \frac{1}{\sin\theta} \partial_\theta( \sin\theta \partial_\theta) + \frac{p_+^2}{\cos^2\theta/2} + \frac{p_-^2}{\sin^2\theta/2} + 2 I_1 c \cos\theta + 2 I_1 d \, \cos^2\theta
\end{equation}
as $L_{CH} \psi = \lambda \psi$ where the eigenvalue $\lambda$ is also called the accessory parameter in the 
context of the Heun equation.
The operator $L_{CH}$ is obtained from the Hamiltonian of the harmonic Lagrange top \eqref{eqn:Ham} by canonical quantisation,
i.e.{} by replacing the kinetic energy with the negative Laplace-Beltrami operator.
The trivial separated equations for $\partial_{\phi_\pm}^2$ 
with periodic boundary conditions are solved and integer values $p_\pm$ are inserted into the remaining operator.
The algebraic form of the equation 
is obtained by introducing $ z  = \cos\theta$ which is the $z$-coordinate of the axis of the top.
The resulting confluent Heun differential operator in algebraic form is
\[
   L_{CH} = - \partial_ z ( ( 1 -  z ^2) \partial_ z ) + \frac{2 p_+^2}{1+ z } + \frac{2 p_-^2}{1 -  z } + 2 I_1 c  z  +2  I_1 d \,  z ^2  \,.
\]
The indices at the regular singular points $z = \pm1$ are  $p_+$ and $p_-$, respectively.
Extending $z = \cos\theta$ to a canonical transformation turns the Hamiltonian \eqref{eqn:HLag} into
\begin{equation} \label{eqn:Hamz}
    H = \frac{1}{2 I_1(t)} \left( ( 1 -  z ^2  ) p_ z ^2 + \frac{2 p_+^2}{1+ z } + \frac{2 p_-^2}{1 -  z }\right) + U( z ) \,.
\end{equation}
Compared to $L_{CH}$ only the first term changes sign, since  $p_+$ and $p_-$ in $L_{CH}$  are already quantum numbers
(or classical actions) and not differential operators any more.
In terms of the original variable $w$ of $P_V$ introducing $z$ amounts to the M\"obius transformation $w =  -(1 +  z )/(1- z )$
that maps the interval $[-1,1]$ in $ z $ to $[0,-\infty]$ in $w$. 
Absorbing $I_1$ into $U$ as before by scaling time we find 
\[
    \frac{d  z }{d \tau} = ( 1 -  z ^2) p_ z , \quad \frac{d p_ z }{d \tau} = -\frac{ \partial H}{\partial  z }
\]
and eliminating $p_ z $ we obtain a version of $P_V$ that is 
the de-quantisation 
of the algebraic form of the generalised spheroidal wave equation
(aka the quantised harmonic Lagrange top), which is
\begin{equation} \label{eqn:PVz}
  \frac{1}{1 -  z ^2} \frac{d^2  z }{d \tau^2} =  \frac{ - z }{(1 -  z ^2)^2} \left(\frac{d z }{d\tau}\right)^2
  + \frac{p_+^2}{(1 +  z )^2} -  \frac{p_-^2}{(1 -  z )^2} - \gamma e^{\tau} - \frac12 \delta e^{2\tau}  z  \,.
\end{equation}
This equation has singularities at $ z  = \pm 1$.
Interestingly, it is also this form that for $\delta = 0$ is most easily mapped to $P_{III}$ \cite{Clarkson05}.


A natural question that arrises is what the actual quantisation of $P_V$ gives. 
Since it is a Hamiltonian system with explicit time-dependence this  leads to a time-dependent Schr\"odinger equation
\[
    i \hbar \frac{\partial}{\partial t} \Psi(\theta, t) = L_{CH} \Psi(\theta, t)
\]
where now the potential in $L_{CH}$ in \eqref{eqn:Heun} has the time-dependence that comes from $P_V$.
This is a $1+1$-dimensional PDE for $\Psi$. 
Some steps in this direction have been taken in \cite{Zabrodin12}.
Interesting connections between quantisation and the Painlev\'e equation are discussed in \cite{Grassi22}.
In \cite{DDN22} we have shown that the quantised Lagrange top, i.e.\ the confluent Heun equation,
has quantum monodromy, which means there is a defect in the joint spectrum of the corresponding commuting 
operators. It would be very interesting to try to understand how this quantum monodromy is connected to 
the iso-monodromy problem associated to $P_V$.

\section{Dynamics on $S^2$}

The motion of the Lagrange top is smooth on $T^*SO(3)$.  Using Euler angles introduces a coordinate singularity
at $\theta = 0, \pi$. This coordinate singularity corresponds to a pole in $P_V$. In this section we are going to 
use the reduction of the Lagrange top to $T^*S^2$ to obtain a global singularity free description
of the dynamics on $S^2$. This can be considered as physically motivated blowup of $P_V$. 
The full symmetry group of the Lagrange top is $S^1\times S^1$, however, there is isotropy
of the group action when the rotation axis are parallel, and hence the fully symmetry reduced system 
is singular at $\theta = 0, \pi$. Only reducing by one of the two $S^1$ symmetries leads to a smooth 
system with two degrees of freedom.

After reduction by the body symmetry the Lagrange top is a Hamiltonian dynamical system on $T^*S^2$.
For more details on the derivation of these equations and the associated Poisson structure see, e.g., \cite{DDN22}.
Denote the axis of the top by $\a \in S^2 \subset \R^3$, $ |\a| = 1$, and by $\l$ the momentum vector in the 
tangent space such that $\l \cdot \a = L_3 = const$. Denote the components of these vectors by 
$(a_x, a_y, a_z)$ and $(l_x, l_y, l_z)$. 
Note that in \eqref{eqn:PVz} the single dependent variable is $a_z \equiv  z $.
The Hamiltonian of the system written in $(\a, \l)$ is 
\[
      H = \frac{1}{2} |\l|^2 + U(a_z)
\]
with equations of motion
\[
      \a' = -\a  \times \l, \quad
      \l' = -\a \times \frac{\partial U}{\partial \a} = -\a \times \e_z U'(a_z) \,.
\]
Here we assume that time has been changed so that $I_1$ is absorbed into $U$, possibly creating time-dependence, and the dash denotes derivatives with respect to the time $\tau$.
In the usual Lagrange top $U$ is linear in $z\equiv a_z$ and hence $U' = c e^\tau$,
or in the harmonic Lagrange top it is $U' = c e^\tau + 2 d a_z e^{2\tau}$.
The case of constant moment of inertia is recovered by setting $\tau = 0$.
The system is invariant under simultaneous rotation of $\a$ and $\l$ about the $z$-axis, and 
the corresponding conserved quantity is $l_z$. Thus after full symmetry reduction the system has one degree of freedom. 
The description presented earlier using Euler angles directly provides this one degree of freedom system. In that notation we have $a_z = \cos\theta$, $\l \cdot \a = L_3 = p_\psi$ and
$l_z = p_\phi$. The problem with Euler angles is that they are singular for $\theta = 0, \pi$
which corresponds to a coordinate singularity in the Euler angles because for these $\theta$ the angles
$\phi$ and $\psi$ are not uniquely defined, but only their sum or difference is.
In $P_V$ the corresponding singularity are $z= \pm1$ in \eqref{eqn:PVz} or at $w = 0$ and $w = -\infty$ in \eqref{eqn:PVw}.
The present description of the Lagrange top as a system on $T^*S^2$ has the advantage
that it provides a natural smooth coordinate system near these singularities.
Note that for real motions $w \le 0$ and in particular the singularity of $P_V$ at $w=1$ does not correspond to a real motion of the real Lagrange top in real time.

Since $l_z$ is constant and $a_z$ is determined through $a_x^2 + a_y^2 + a_z^2 = 1$ we can 
project the equations onto the $xy$-components
and write it  in complex form with $a = a_x + i a_y$ and $l = l_x + i l_y$ as (a deceptively linear looking) non-linear system on $\mathbb{C}^2$
\begin{equation} \label{eqn:C2sys}
   \begin{pmatrix}  a' \\  l' \end{pmatrix} = 
   i \begin{pmatrix}
   	-l_z  & a_z  \\ -U'(a_z)  & 0 
   \end{pmatrix} 
   \begin{pmatrix}
   	a \\ l
   \end{pmatrix} \,.
\end{equation}
This system of ODEs has an equilibrium point at the origin, which 
corresponds to the north- or south-pole of the sphere.
Linearisation about this equilibrium amounts to setting $a_z = \pm 1$. 
We keep $a_z$ in the equation to treat both signs simultaneously.
The resulting 2nd order linear equation is
\[
    a'' + i l_z a' - a_z U'(a_z) a = 0, \quad 
    l_z = const, \, U'(a_z) = c e^\tau + 2 d e^{2\tau} a_z,  \, a_z = \pm 1 \,.
\]
Returning to the original time $t = e^\tau$ we find $a' = t \dot a$ and $a'' = t^2 \ddot a + t \dot a$ and after
cancelling an overall factor of $t$ 
\begin{equation} \label{eqn:axiay}
    t \ddot a + ( 1 + i l_z) \dot a  - a_z (c + 2d t a_z) a = 0\,.
\end{equation}
For $\delta = 0$ this is the Bessel equation, while in general it is the confluent hypergeometric equation.
If we remove the time-dependence in the equation by setting $t = 1$ the linear equation 
describes the Hopf bifurcation by which the sleeping top is de-stabilised when the spin rate
$l_z$ becomes too slow, see, e.g., \cite{DDN22}. With time-dependent moment of inertia 
passing the stability threshold  results in the onset of oscillations. 


Solutions that are interesting from a physical point of view are those that approach $a_z \equiv z = \pm 1$ for 
$\tau \to \pm\infty$. The blow up of $P_V$ near singularities has been studied in \cite{Joshi18}.
Adding the non-linear term  $ia_z' l$  to  \eqref{eqn:axiay} where now $a_z = \pm \sqrt{ 1 - a \bar a}$ and $l$ is expressed 
in terms of $a$ and its derivative using \eqref{eqn:C2sys}, which gives
\[
  t \ddot a + (1  + i l_z ) \dot a  -  a_z (c + 2 d t a_z) a = i \frac{ a \dot{\bar a} + \bar a \dot a}{ 2 a_z^2} ( t\dot a + i l_z a)
\]
It would be interesting to study how this equation compares to the blown up $P_V$. 
The main advantage of the equation when written in $a = a_x + i a_y$
instead of $a_z$ is that it is regular near $a_z = \pm 1$. 
There is, however, a square root in the equation because
$a_z = \pm \sqrt{ 1 - a \bar a}$. 

We conclude with a qualitative discussion of solutions of $P_V$ corresponding to the real Lagrange top with 
time-dependent moment of inertia. 
It appears that the parameters relevant for this are $\alpha \le 0$, $\beta \ge 0$, $\gamma > 0$, $\delta = 0$.
For $\delta > 0$ (i.e.\ with the extra harmonic terms in the top) this is the class of solutions studied in \cite{McLeod99}.
In section 3 we gave a quick review of the properties of solutions of the time-independent Lagrange top.
What changes with the time dependence? 
The simplest case of the pendulum with time-dependent length occurs for $p_+= p_- = 0$. 
For $p_\theta = 0$ there are two equilibrium solutions at $z = \pm 1$, the minimum and 
the maximum of the potential. Now consider non-zero $p_\theta$.
Starting at $\tau = -\infty$ in this case $\theta$ increases linearly with time with slope 
given by $p_\theta$. When $\tau$ crosses towards positive times the potential becomes important,
and for $\tau \to +\infty$ the solution spirals to a potential minimum with $\theta = (2n - 1) \pi$ for some integer $n$.
While spiralling towards the minimum the energy goes to $-\infty$, since $\cos\theta \to -1$ and it 
is multiplied by an exponentially growing term.
Increasing the initial $p_\theta$  the solution will eventually change from ``basin'' $n$ to basin $n+1$.
By continuity between these lies a unique solution with a particular $p_\theta$ that will asymptote
to the potential maximum with $\theta = 2n \pi$.
On a  qualitative level the behaviour is like a pendulum with friction, but the physical process 
(and the details of the solution) are of course very different. 
Nevertheless, in both systems the exceptional solutions that approach the unstable maximum 
for $\tau \to +\infty$ exist.
Now we are going to discuss solutions where at least one $p_\pm$ is non-zero.
We are going to discuss the limit $\tau \to -\infty$ and $\tau \to +\infty$ in turns.

For $\tau \to -\infty$ the potential terms vanish, and the dynamics is free motion on $SO(3)$. 
Considering the double cover $S^3$ this implies that the solutions are great circles on $S^3$ (recall that the term 
proportional to $p_\psi^2$ in the Hamiltonian has no counterpart in $P_V$). 
Hence $z$ will oscillate between a minimum and a maximum which depend on the values of $p_\pm$. 
The only solutions that do not oscillate in this limit correspond to the great circle that 
has $z = 0$. This solution is possible only when $p_+ p_- = 0$. 

When $\tau$ reaches the vicinity of 0 the system starts to behave like the Lagrange top. 
This regime is short-lived unless all parameters are large.
Eventually for $\tau \to +\infty$ the potential dominates the Hamiltonian. 
As for the pendulum most solutions approach the potential minimum $z = -1$ in this limit.
In the time-independent Lagrange top $z = -1$ is only accessible when the conserved momentum satisfies $p_+ = 0$,
because otherwise the energy diverges, which is a contradiction to energy conservation.
However, in the time-dependent case the energy is not constant, and in fact $\dot E = \partial H / \partial \tau = \gamma e^\tau z$
which is negative for negative $z$. Thus the system will loose energy and the solutions approach $z = -1$ in an oscillatory manner. 

A different class of interesting solutions are those that approach the upright sleeping top with $z = 1$ for $\tau \to \infty$.
Solutions for which $z \equiv 1$ certainly exists but cannot be seen in $P_V$ because
of the singularity of the equation at $z =1$. 
However, for dynamics on $S^2$ the vectors $\a = (0,0,1)$ and $\l = (0,0,l_z)$ clearly correspond to that equilibrium solution. 
Can this solution be approached from $z < 1$? In the time-independent case the answer
is yes if $p_- = 0$ and the sleeping top is unstable (i.e. $l_z$ is not too large), 
in which case the equilibrium has a stable manifold along which it can be approached.
With time-dependence for $\tau \to \infty$ this will be harder, but by a continuity argument similar to that applied to 
the pendulum this is possible at least when $p_- = 0$.
Thus the most special solutions of $P_V$ related to real motions of the time-dependent Lagrange top are those that 
connect $z = 0$ at  $\tau = -\infty$ to $z = 1$ at $\tau = +\infty$ without any oscillations.

\section{$P_V$ on an orbifold}

The full symmetry reduction of the Lagrange top by both its $S^1$ symmetries leads to a Poisson structure in $\mathbb{R}^3$
whose Casimir defines a smooth non-compact surface for most values of $p_\pm$, which becomes an orbifold
when either $p_+$ or $p_-$ vanishes. The singularity appears because the $S^1 \times S^1$ action is not free but 
has isotropy exactly for the sleeping tops for which $p_+$ or $p_-$ vanishes.
In the following we are going to describe this orbifold and its regularisation / blow-up.
This will allow for a smooth description of motion at and near $w = 0$ and $w=\infty$ for arbitrary time.

The dynamics on $S^2$ with  rotational symmetry around the $z$-axis is best described using 
complex variables $a = a_x + i a_y$, $l = l_x + i l_y$. The $S^1$ action in these variables is 
simply multiplication $(a, l) \mapsto (a e^{i\phi}, l e^{i \phi})$ and the invariants of the $S^1$ action
are $a \bar a \ge 0$, $T = l \bar l \ge 0$, and the complex $a \bar l = u + i v$. These invariants satisfy 
the relation $u^2 + v^2 = | a \bar l |^2 = a \bar a \, T$. The trivial invariants $z$ and $l_z$ 
are related to these invariants through $z^2 + a \bar a = 1$ and $ z l_z + u = \a \cdot \l = L_3$.
Using these to eliminate $u$ and $a \bar a$ in the relation gives the cubic Casimir
\[
     C(T,z,v) =  (L_3 - z l_z)^2 + v^2 - (1 - z^2) T = 0 
\]
and the Hamiltonian
\[
    H(T,z) = \tfrac12 T + U(z)\,.
\]
The Poisson structure is given by taking the cross product with the gradient of $C$.
The zero-level of the Casimir defines a surface which is the reduced phase space. It is a  non-compact surface.
It is smooth unless $L_3 \pm l_z = 0$. When $L_3 \pm l_z = 0$ then the reduced phase space is  
an orbifold with singular point $z = \mp 1$, $v = 0$. We are now going to show that these singular 
points are indeed conical singularities. 

From now on $L_3 = \mp l_z$.
Firstly, translate the singular point to the origin, $z = \mp 1 \pm \Delta z$, such 
that the Casimir becomes $l_z^2 \Delta z^2 - 2 T \Delta z + v^2  + T \Delta z^2$.
Both singular points at $z = -1 +\Delta z$ and $z = 1 - \Delta z$ lead to the same Casimir.
Secondly, rotate the $(T, \Delta z)$ plane so that the Hessian at the origin (which is the singular point) is diagonal.
Thirdly, scale the new coordinates so that the eigenvalues of the Hessian at the origin are equal in magnitude.
Together this gives an affine area-preserving transformation of $(T,z)$ to new coordinates $(X,Y)$ such that 
the Casimir is
\[
    \tilde C(X,Y,v) = -X^2 + Y^2 + v^2 + (X+Y)^2 (X \lambda_+ + Y \lambda_-) (4 + l_z^4)^{-3/4}
\]
where $2 \lambda_\pm = l_z^2 \pm \sqrt{4 + l_z^4}$ so that $\lambda_+ \lambda_- = -1$.  
The quadratic terms describe the conical singularity at the origin.
The cone can be ``unrolled'' onto the plane by introducing polar coordinates for $(Y,v)$ where $X$ is the radius 
and then doubling the angle. At quadratic order this amounts to introducing new cartesian coordinates $Y + i v = (\tilde Y + i \tilde v)^2/r = (\tilde Y^2 - \tilde v^2 + 2 i \tilde v \tilde Y)/r$ and $X = r$ where $r^2 =  \tilde Y^2 + \tilde v^2$.

This process gives an equation that is equivalent to a double cover of the real $P_V$ near the singular points $ w = 0$, $w = \infty$.
The main difference to the equation in the previous section is that there we had a complex 2nd order equation
corresponding to real solutions of the only partially symmetry reduced Lagrange top. 
By contrast, the conical singularity of the Poisson structure leads to a single real 2nd order equation 
that corresponds to real solutions of the fully symmetry reduced Lagrange top. 
The additional dimensions in the previous section were a consequence of the fact that there we 
did not consider the fully symmetry reduced Lagrange top.


%

 \printbibliography

\end{document}